\documentclass[journal,10pt,twocolumn,twoside]{IEEEtran}
\usepackage{graphicx} 
\usepackage{subfigure} 
\usepackage{amsmath}
\usepackage{amssymb}
\usepackage{amsxtra}%
\usepackage{amsfonts}%
\usepackage{amsthm}
\usepackage{cite}
\usepackage{ bbold }

\usepackage{algorithm}
\usepackage{algorithmic}

\usepackage{hyperref}

\newcommand{\Pset}{\mathcal{P}}

\newcommand{\R}{\mathbb{R}}

\newcommand{\F}{\mathcal{F}}

\newtheorem{proposition}{Proposition}
\newtheorem{definition}{Definition}
\newtheorem{theorem}{Theorem}

\newcommand{\N}{\mathbb{N}}
\newcommand{\X}{\mathcal{X}}
\newcommand{\Y}{\mathcal{Y}}

\begin{document}

\title{Pareto-depth for Multiple-query Image Retrieval}

\author{Ko-Jen Hsiao, Jeff Calder, \IEEEmembership{Member,~IEEE},
	and~Alfred O.~Hero III, \IEEEmembership{Fellow,~IEEE}
\thanks{This work was partially supported by ARO grant W911NF-09-1-0310. The paper is submitted to IEEE Transaction on Image Processing on February 20, 2014.
K.-J. Hsiao and A. Hero are with the Department of Electrical Engineering and Computer Science, University of Michigan, Ann Arbor. Jeff Calder is with the Department of Mathematics, University of Michigan, Ann Arbor. (email: \{coolmark,hero,jcalder\}@umich.edu.)}}

\maketitle

\begin{abstract}
  Most content-based image retrieval systems consider either one single query, or multiple queries that include the same object or represent the same semantic information. In this paper we consider the content-based image retrieval problem for multiple query images corresponding to \emph{different} image semantics.  We propose a novel multiple-query information retrieval algorithm that combines the Pareto front method (PFM) with efficient manifold ranking (EMR).  We show that our proposed algorithm outperforms state of the art multiple-query retrieval algorithms on real-world image databases.  We attribute this performance improvement to concavity properties of the Pareto fronts, and prove a theoretical result that characterizes the asymptotic concavity of the fronts. \end{abstract}

\section{Introduction}
In the past two decades content-based image retrieval (CBIR) has become an important problem in machine learning and information retrieval~\cite{frome2007image,gia2004instance,vasconcelos1999learning}.  
Several image retrieval systems for multiple queries have been proposed in the literature~\cite{belkin1993effect,arandjelovic2012multiple,jin2005improving}. In most systems, each query image corresponds to the same image semantic concept, but may possibly have a different background, be shot from an alternative angle, or contain a different object in the same class.  The idea is that by utilizing multiple queries of the same object, the performance of single-query retrieval can be improved. We will call this type of multiple-query retrieval \emph{single-semantic-multiple-query} retrieval. Many of the techniques for single-semantic-multiple-query retrieval involve combining the low-level features from the query images to generate a single averaged query~\cite{arandjelovic2012multiple}.

In this paper we consider the more challenging problem of finding images that are relevant to multiple queries that represent different image semantics. In this case, the goal is to find images containing relevant features from \emph{each and every} query.  Since the queries correspond to different semantics,  desirable images will contain features from several distinct images, and will not necessarily be closely related to any individual  query.  This makes the problem fundamentally different from single query retrieval, and from single-semantic-multiple-query retrieval.    In this case, the query images will not have similar low level features, and forming an averaged query is not as useful.

Since relevant images do not necessarily have features closely aligned with any particular query, many of the standard retrieval techniques are not useful in this context.  For example, bag-of-words type approaches, which may seem natural for this problem, require the target image to be closely related to several of the queries.  Another common technique is to input each query one at a time and average the resulting similarities.  This tends to produce images closely related to one of the queries, but rarely related to all at once.   Many other multiple-query retrieval algorithms are designed specifically for the single-semantic-multiple-query problem~\cite{arandjelovic2012multiple}, and again tend to find images related to only one, or a few, of the queries.  

Multiple-query retrieval is related to the metasearch problem in computer science.  In metasearch, the problem is to combine search results for the same query across multiple search engines.  This is similar to the single-semantic-multiple-query problem in the sense that every search engine is issuing the same query (or semantic). Thus, metasearch algorithms are not suitable in the context of multiple-query retrieval with several distinct semantics. 

In this paper, we propose a novel algorithm for multiple-query image retrieval that combines the Pareto front method (PFM) with efficient manifold ranking (EMR).  The first step in our PFM algorithm is to issue each query individually and rank all samples in the database based on their dissimilarities to the query. Several methods for computing representations of images, like SIFT and HoG, have been proposed in the computer vision literature, and any of these can be used to compute the image dissimilarities. Since it is very computationally intensive to compute the dissimilarities for every sample-query pair in large databases, we use a fast ranking algorithm called Efficient Manifold Ranking (EMR) \cite{xu2011} to compute the ranking without the need to consider all sample-query pairs.  EMR can efficiently discover the underlying geometry of the given database and significantly reduces the computational time of traditional manifold ranking. Since EMR has been successfully applied to single query image retrieval, it is the natural ranking algorithm to consider for the multiple-query problem.

The next step in our PFM algorithm is to use the ranking produced by EMR to create Pareto points, which correspond to dissimilarities between a sample and every query. Sets of Pareto-optimal points, called \emph{Pareto fronts}, are then computed. The first Pareto front (depth one) is the set of non-dominated points, and it is often called the \emph{Skyline} in the database community. The second Pareto front (depth two) is obtained by removing the first Pareto front, and finding the non-dominated points among the remaining samples. This procedure continues until the computed Pareto fronts contain enough samples to return to the user, or all samples are exhausted. The process of arranging the points into Pareto fronts is called \emph{non-dominated sorting}.
\begin{figure*}[ht]
\vskip 0.2in
\begin{center}
{\includegraphics[trim=0cm 12.6cm 0cm 12.6cm, clip=true, width=18cm]{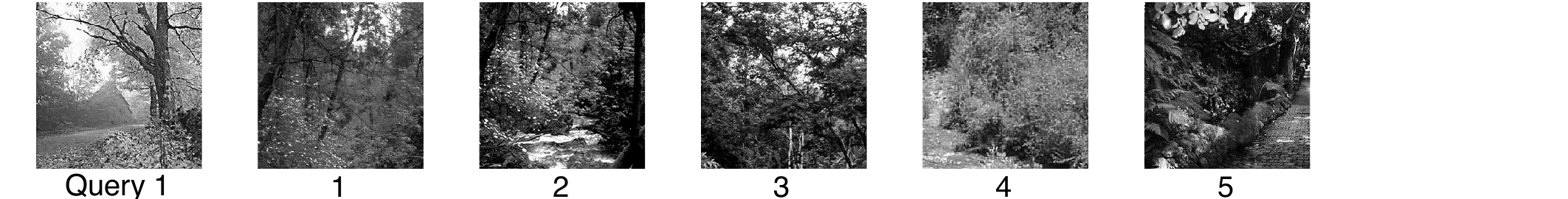}}\\
{\includegraphics[trim=0cm 12.6cm 0cm 12.6cm, clip=true, width=18cm]{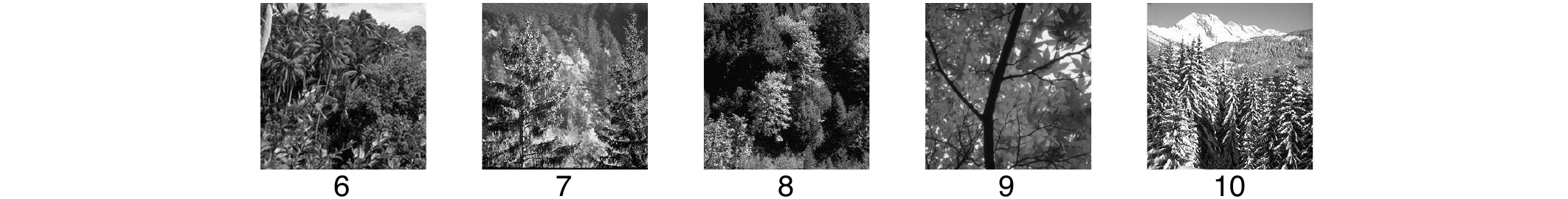}}\\
{\includegraphics[trim=0cm 12.6cm 0cm 12.6cm, clip=true, width=18cm]{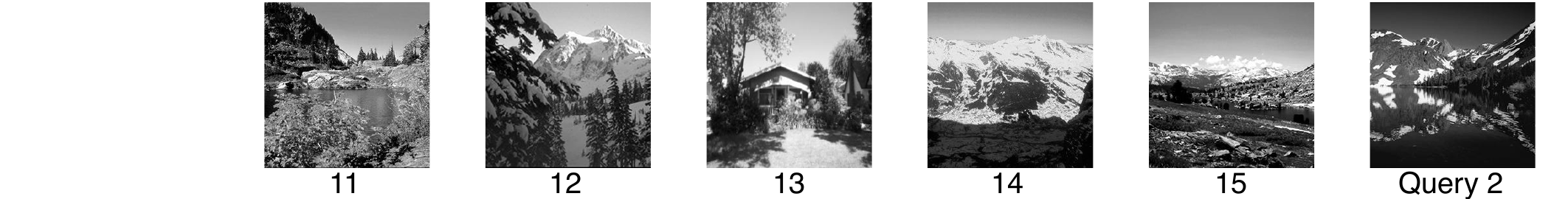}}
\caption{Images located on the first Pareto front when a pair of query images are issued. Images from the middle part of the front (images 10, 11 and 12) contain semantic information  from \emph{both} query images. The images are from {\bf Stanford 15 scene} dataset. }
\label{ImageSeq}
\end{center}
\vskip -0.2in
\end{figure*} 

A key observation in this work is that the middle of the Pareto front is of fundamental importance for the multiple-query retrieval problem. As an illustrative example, we show in Figure \ref{ImageSeq} the  images from the first Pareto front for a pair of query images corresponding to a forest and a mountain. The images are listed according to their position within the front, from one tail to the other.  The images located at the tails of the front are very close to one of the query images, and may not necessarily have any features in common with the other query.  However, as seen in Figure \ref{ImageSeq}, images in the middle of the front (e.g., images 10, 11 and 12) contain relevant features from \emph{both} queries, and hence are very desirable for the multiple-query retrieval problem.  It is exactly these types of images that our algorithm is designed to retrieve. 

The Pareto front method is well-known to have many advantages when the Pareto fronts are non-convex~\cite{ehrgott2005}. In this paper, we present a new theorem that characterizes the asymptotic convexity (and lack thereof) of Pareto fronts as the size of the database becomes large.  This result is based on establishing a connection between Pareto fronts and chains in partially ordered finite set theory.  The connection is as follows: a data point is on the Pareto front of depth $n$ if and only if it admits a maximal chain of length $n$.   This connection allows us to utilize results from the literature on the longest chain problem, which has a long history in probability and combinatorics.  Our main result (Theorem \ref{thm:macro}) shows that the Pareto fronts are asymptotically convex when the dataset can be modeled as \emph{i.i.d.}~random variables drawn from a continuous separable log-concave density function $f:[0,1]^d \to (0,\infty)$.  This theorem suggests that our proposed algorithm will be particularly useful when the underlying density is not log-concave.  We give some numerical evidence (see Figure \ref{basicReal}) indicating that the underlying density is typically not even quasi-concave.  This helps to explain the performance improvement obtained by our proposed Pareto front method.

We also note that our PFM algorithm could be applied to automatic image annotation of large databases. Here, the problem is to automatically assign keywords, classes or captioning to images in an unannotated or sparsely annotated database. Since images in the middle of first few Pareto fronts are relevant to all queries, one could issue different query combinations with known class labels or other metadata, and automatically annotate the images in the middle of the first few Pareto fronts with the metadata from the queries.  This procedure could, for example, transform a single-class labeled image database into one with multi-class labels. Some major works and introductions to automatic image annotation can be found in \cite{jeon2003automatic, carneiro2007supervised, russell2008labelme}.

The rest of this paper is organized as follows. We discuss related work in Section \ref{related}. In Section \ref{PFM}, we introduce the Pareto front method and present a theoretical analysis of the convexity properties of Pareto fronts. In Section \ref{ParetoIR} we show how to apply the Pareto front method (PFM) to the multiple-query retrieval problem and briefly introduce Efficient Manifold Ranking. Finally, in Section \ref{experiment} we present experimental results and demonstrate a graphical user interface (GUI) that allows the user to explore the Pareto fronts and visualize the partially ordered relationships between the queries and the images in the database.

\section{Related work}
\label{related}

\subsection{Content-based image retrieval}

Content-based image retrieval~(CBIR) has become an important problem over the past two decades. Overviews can be found in \cite{liu2007survey, datta2008image}. A popular image retrieval system is query-by-example (QBE) \cite{zloof1975query,hirata1992query}, which retrieves images relevant to one or more queries provided by the user. In order to measure image similarity, many sophisticated color and texture feature extraction algorithms have been proposed; an overview can be found in \cite{liu2007survey,datta2008image}. SIFT \cite{lowe2004} and HoG \cite{dalal2005histograms} are two of most well-known and widely used feature extraction techniques in computer vision research. Several CBIR techniques using multiple queries have been proposed \cite{arandjelovic2012multiple,jin2005improving}. Some methods combine the queries together to generate a query center, which is then modified with the help of relevance feedback.  Other algorithms issue each query individually to introduce diversity and gather retrieved items scattered in visual feature space \cite{jin2005improving}. 

The problem of ranking large databases with respect to a  similarity measure has drawn great attention in the machine learning and information retrieval fields. Many approaches to ranking have been proposed, including learning to rank \cite{liu2009learning,burges2005learning}, content-based ranking models (BM25, Vector Space Model), and link structure ranking model \cite{brin1998anatomy}. Manifold ranking \cite{zhou2004ranking,zhou2004learning} is an effective ranking method that takes into account the underlying geometrical structure of the database. Xu et al.~\cite{xu2011} introduced an algorithm called Efficient Manifold Ranking (EMR) which uses an anchor graph to do efficient manifold ranking that can be applied to large-scale datasets. In this paper, we use EMR to assign a rank to each sample with respect to each query before applying our Pareto front method.

\subsection{Pareto method}

There is wide use of Pareto-optimality in the machine learning community \cite{jin2008pareto}. 
Many of these methods must solve complex multi-objective optimization problems, where finding even the first Pareto front is challenging.  
Our use of Pareto-optimality differs as we generate \emph{multiple} Pareto fronts from a finite set of items, and as such we do not require sophisticated methods to compute the fronts.

In computer science the first Pareto front, which consists of the set of non-dominated points, is often called the Skyline. Several sophisticated and efficient algorithms have been developed for computing the Skyline \cite{borzsony2001skyline, kossmann2002shooting, tan2001efficient, papadias2003optimal}. Various Skyline techniques have been proposed for different applications in large-scale datasets, such as multi-criteria decision making, user-preference queries, and data mining and visualization \cite{hristidis2001prefer,jin2004mining,agrawal2000framework}. Efficient and fast Skyline algorithms \cite{kossmann2002shooting} or fast non-dominated sorting \cite{jensen2003} can be used to find each Pareto front in our PFM algorithm for large-scale datasets. 

Sharifzadeh and Shahabi\cite{sharifzadeh2006spatial} introduced Spatial Skyline Queries (SSQ) which is similar to the multiple-query retrieval problem.  However, since EMR is not a metric (it doesn't satisfy the triangle inequality), the relation between the first Pareto front and the convex hull of the queries, which is exploited by Sharifzadeh and Shahabi\cite{sharifzadeh2006spatial}, does not hold in our setting.  Our method also differs from SSQ and other Skyline research because we use multiple fronts to rank items instead of using only Skyline queries.  We also address the problem of combining EMR with the Pareto front method for multiple queries associated with \emph{different} concepts, resulting in \emph{non-convex} Pareto fronts.  To the best of our knowledge, this problem has not been widely researched. 

A similar Pareto front method has been applied to the gene ranking problem \cite{hero2004pareto}. Their approach utilized Pareto methods to rank genes based on multiple criteria of interest to a biologist.  In another related work, Hsiao et al.~\cite{hsiao2012} proposed a multi-criteria anomaly detection algorithm utilizing Pareto depth analysis. This approach uses multiple Pareto fronts to define a new dissimilarity between samples based on their Pareto depth. In their case, each Pareto point corresponds to a similarity vector between  pairs of database entries under muitiple similarity criteria. In this paper, a Pareto point corresponds to a vector of dissimilarities between a single entry in the database and multiple queries. 

A related field is metasearch \cite{aslam2001models, meng2002building}, in which one query is issued in different systems or search engines, and  different ranking results or scores for each item in the database are obtained. These different scores are then combined to generate one final ranked list. Many different methods, such as Borda fuse and CombMNZ, have been proposed and are widely used in the metasearch community.  The same methods have also been used to combine results for different representations of a query\cite{lee1997analyses,belkin1993effect}. However these algorithms are designed for the case that the queries represent the same semantics. In the multiple-query retrieval setting this case is not very interesting as it can easily be handled by other methods, including linear scalarization. 

In contrast we study the problem where each query corresponds to a different image concept.  In this case metasearch methods are not particularly useful, and  are significantly outperformed by the Pareto front method.  For example Borda fusion gives higher rankings to the tails of the fronts, and thus is similar to linear scalarization. CombMNZ gives a higher ranking to documents that are relevant to multiple-query aspects, but it utilizes a sum of all document scores, and as such is intimately related to linear scalarization with equal weights, which is equivalent to the Average of Multiple Queries (MQ-Avg) retrieval algorithm~\cite{arandjelovic2012multiple}. We show in Section \ref{experiment} that our Pareto front method significantly outperforms MQ-Avg, and all other multiple-query retrieval algorithms.

Another related field is multi-view learning \cite{blum1998combining,sindhwani2005co}, in which data is represented by multiple sets of features that are referred to as ``views''.  Training in one view will usually improve the learning in another, although there is often view disagreement, in which the same sample may belong to a different class in each view \cite{christoudias2012multi}. The views are similar to criteria in our problem setting. However, different criteria may be orthogonal and could even give contradictory information; hence there may be severe view disagreement, and training in one view could actually worsen performance in another view.  A similar area is that of multiple kernel learning \cite{gonen2011multiple}, which is typically applied to supervised learning problems instead of retrieval problems.

\section{Pareto Front method}
\label{PFM}

Pareto-optimality is a powerful concept that has been applied in many fields, including economics, computer science, and the social sciences \cite{ehrgott2005}.  We give here a brief overview of Pareto-optimality and define the notion of a Pareto front.

In the general setting of a discrete multi-objective optimization problem, we have a finite set $\mathcal{S}$ of feasible solutions, and $T$ criteria $f_1,\dots,f_T:  \mathcal{S}\to \R$ for evaluating the feasible solutions. One possible goal is to find $x \in \mathcal{S}$ minimizing \emph{all} criteria simultaneously. In most settings, this is an impossible task.   Many approaches to the multi-objective optimization problem reduce to combining all $T$ criteria into one.  When this is done with a linear combination, it is usually called \emph{linear scalarization} \cite{ehrgott2005}.  Different choices of weights in the linear combination yield different minimizers.  Without prior knowledge of the relative importance of each criterion, one must employ a grid search over all possible weights to identify a set of feasible solutions. 

A more robust and principled approach involves finding the Pareto-optimal solutions.  We say that a feasible solution $x \in \mathcal{S}$ is \emph{Pareto-optimal} if no other feasible solution ranks better in every objective.  More precisely, we say that $x$ \emph{strictly dominates} $y$ if $f_i(x) \leq f_i(y)$ for all $i$, and $f_j(x) < f_j(y)$ for some $i$.  An item $x\in \mathcal{S}$ is Pareto-optimal if it is not strictly dominated by another item.  The collection of Pareto-optimal feasible solutions is called the first \emph{Pareto front}.  It contains all solutions that can be found via linear scalarization, as well as other items that are missed by linear scalarization.   The first Pareto front is denoted by $\F_1$.   The second Pareto front, $\F_2$, is obtained by removing the first Pareto front from $\mathcal{S}$ and finding the Pareto front of the remaining data.  More generally, the $i^{\rm th}$ Pareto front is defined by
$$\F_i=\text{Pareto front of the set }
\mathcal{S} \setminus \left(\bigcup_{j=1}^{i-1}\F_j\right).$$
If $x \in \F_k$ we say that $x$ is at a \emph{Pareto depth} of $k$.  We say that a Pareto front $\F_i$ is \emph{deeper} 
than $\F_j$ if $i>j$. 

\begin{figure*}[!t]
\centering
\subfigure[]{\includegraphics[width=0.4\textwidth]{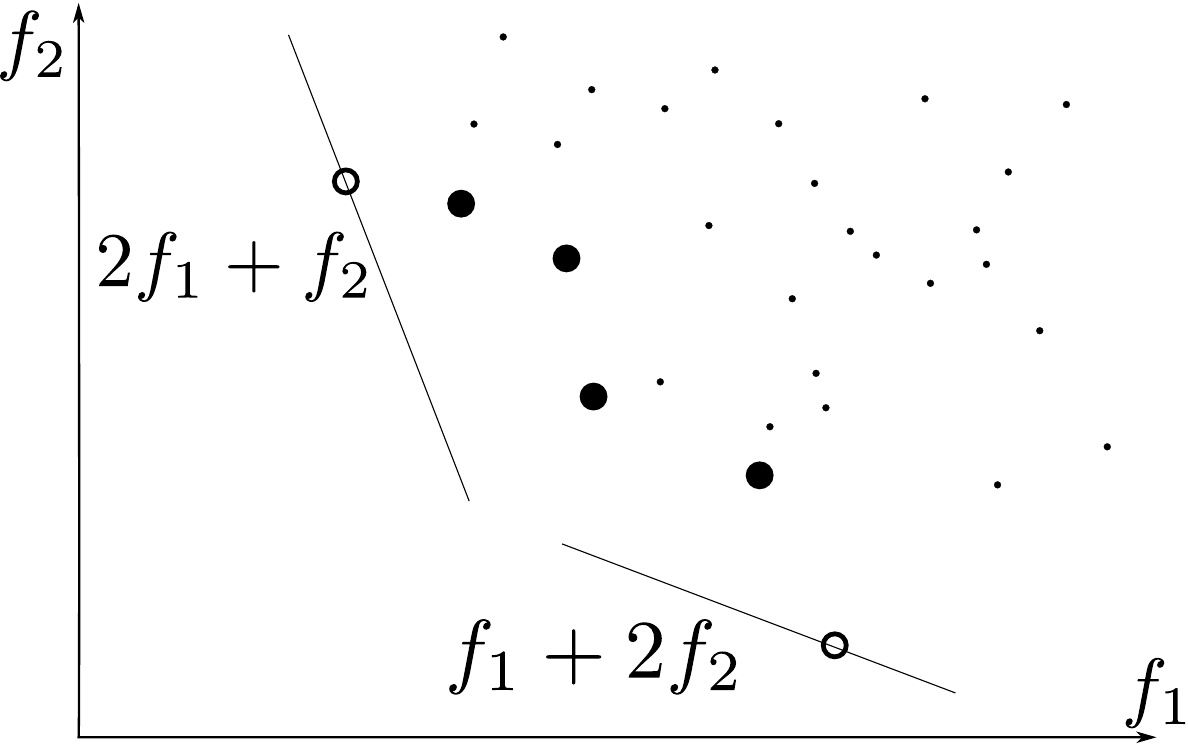}\label{basicExample}}
\hspace{0.75in}
\subfigure[]{\includegraphics[width=0.32\textwidth]{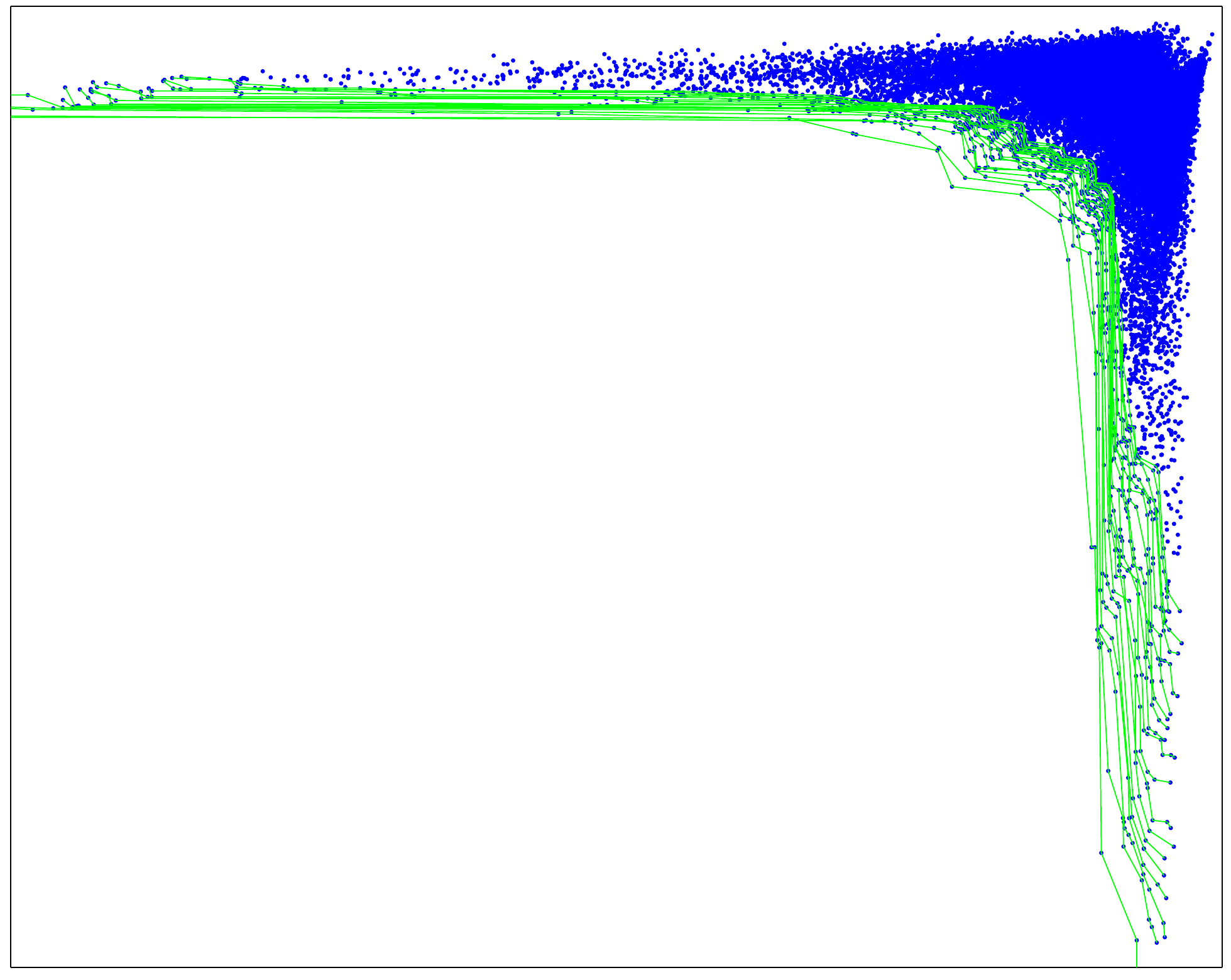}\label{basicReal}}
\caption{(a) Depiction of nonconvexities in the first Pareto front.  The large points are Pareto-optimal, having the highest Pareto ranking by criteria $f_1$ and $f_2$, but only the hollow points can be obtained by linear scalarization. Here $f_1$ and $f_2$ are the dissimilarity values for query 1 and query 2, respectively. (b) Depiction of nonconvexities in the Pareto fronts in the real-world {\bf Mediamill} dataset used in the experimental results in Section \ref{experiment}. The points on the non-convex portions of the fronts will  be retrieved later by any scalarization algorithm, even though they correpsond to equally good images for the retrieval problem. }
\end{figure*}
\setcounter{subfigure}{0}

The simple example in Figure \ref{basicExample} shows the advantage of using Pareto front methods for ranking. Here the number of criteria is $T =2$ and the Pareto points $[f_1(x), f_2(x)]$, for $x \in S$, are shown in Figure \ref{basicExample}. In this figure the large points are Pareto-optimal, but only the hollow points can be obtained as top ranked items using linear scalarization. It is well-known, and easy to see in Figure \ref{basicExample}, that linear scalarization can only obtain Pareto points on the boundary of the convex hull of the Pareto front.  The same observation holds for deeper Pareto fronts.  Figure \ref{basicReal} shows Pareto fronts for the multiple-query retrieval problem using real data from the \textbf{Mediamill} dataset, introduced in Section \ref{experiment}.  Notice the severe non-convexity in the shapes of the real Pareto fronts in Figure \ref{basicReal}.  This is a key observation, and is directly related to the fact that each query corresponds to a differnt image semantic, and so there are no images that are very closely related to both queries.

\subsection{Information retrieval using Pareto fronts}
\label{IRusingPRM}
In this section we introduce the  Pareto front method for the multiple-query information retrieval problem. Assume that a dataset $\mathcal{X}_N=\{X_1,\ldots,X_N\}$ of data samples is available. Given a query $q$, the objective of retrieval is to return samples that are related to the query. When multiple queries are present, our approach issues each query individually and then combines their results into one partially ordered list of Pareto-equivalent retrieved items at successive Pareto depths.
For $T>1$, denote the $T$-tuple of queries by $\{q_1, q_2,...,q_T\}$ and the dissimilarity between $q_i$ and the $j^{\rm th}$ item in the database, $X_j$, by $d_i(j)$. For convenience, define $d_i\in\mathbb{R}_{+}^N$ as the dissimilarity vector between $q_i$ and all samples in the database.
Given $T$ queries, we define a \emph{Pareto point} by $P_{j}=[d_1(j), \ldots, d_T(j)] \in \mathbb{R}_{+}^T, j \in \{1,\ldots,N\}$. Each Pareto point $P_j$ corresponds to a sample $X_j$ from the dataset $\mathcal{X}_N$.  For convenience, denote the set of all Pareto points by $\Pset$. By definition, a Pareto point $P_i$ strictly dominates another point $P_j$ if $d_l(i) \leq d_l(j)$  for all $l \in \{1, \ldots, T\}$ and $d_l(i) < d_l(j)$ for some $l$.  One can easily see that if $P_i$ dominates $P_j$, then $X_i$ is closer to every query than $X_j$. Therefore, the system should return $X_i$ before $X_j$. 
The key idea of our approach is to return samples corresponding to which Pareto front they lie on, i.e., we return the points from $\F_1$ first, and then $\F_2$, and so on until a sufficient number of images have been retrieved.  Since our goal is to find images related to each and every query, we start returning samples from the middle of the first Pareto front and work our way to the tails.  The details of our algorithm  are presented in Section \ref{ParetoIR}.

\subsection{Properties of Pareto fronts}
\label{sec:theory}
Previous works have studied the distribution of the number of Pareto-optimal points missed by linear scalarization.  Hsiao et al.\cite{hsiao2012} prove two theorems characterizing how many Pareto-optimal points are missed, on average and asymptotically, due to nonconvexities in the geometry of the Pareto point cloud, called \emph{large-scale} non-convexities, and nonconvexities due to randomness of the Pareto points, called \emph{small-scale} nonconvexities.  In particular, Hsiao et al.\cite{hsiao2012} show that even when the Pareto point cloud appears convex, at least $1/6$ of the Pareto-optimal points are missed by linear scalarization in dimension $T=2$.  

We present here some new results on the asymptotic convexity of Pareto fronts.  Let $X_1,\dots,X_n$ be \emph{i.i.d.}~random variables on $[0,1]^d$ with probability density function $f:[0,1]^d\to \R$ and set $\X_n =\{X_1,\dots,X_n\}$.  Then $(\X_n,\leqq)$ is a partially ordered set, where $\leqq$ is the usual partial order on $\R^d$ defined by
\[x \leqq y \iff x_i \leq y_i \  {\rm for \ all \ } i\in \{1,\dots,d\}.\]
Let $\F_1,\F_2,\dots$ denote the Pareto fronts associated with $\X_n$, and let $h_n:[0,1]^d \to \R$ denote the Pareto depth function defined by
\begin{equation}
h_n(x) = \max \{i \in \N \, : \, \F_i \leqq x \},
\end{equation}
where for simplicity we set $\F_0 = \{(-1,\dots,-1)\}$, and we write $\F_i \leqq x $ if there exists $y \in \F_i$ such that $y  \leqq x$.
The function $h_n$ is a (random) piecewise constant function that ``counts'' the Pareto fronts associated with $X_1,\dots,X_n$. 

Recall that a \emph{chain} of length $\ell$ in $\X_n$ is a sequence $x^1,\dots,x^\ell \in \X_n$ such that $x^1 \leqq x^2 \leqq \cdots \leqq x^\ell$. Define $u_n:[0,1]^d \to \R$ by
\begin{equation}
u_n(x) = \max \{\ell \in \N \, : \, \exists \ x^1 \leqq \cdots \leqq x^\ell \leqq x {\rm \ in \ } \X_n\}.
\end{equation}
The function $u_n(x)$ is the length of the longest chain in $\X_n$ with maximal element $x_\ell \leqq x$.   We have the following alternative characterization of $h_n$:
\begin{proposition}\label{prop:chain}
$h_n(x) = u_n(x)$ with probability one for all $x \in [0,1]^d$.
\end{proposition}
\begin{proof}
Suppose that $X_1,\dots,X_n$ are distinct.  Then each Pareto front consists of mutually incomparable points.  Let $x \in [0,1]^d$, $r = u_n(x)$ and $k=h_n(x)$.  By the definition of $u_n(x)$, there exists a chain $x_1 \leqq \cdots \leqq x_r$ in $\X_n$ such that $x_r \leqq x$.  Noting that each $x_i$ must belong to a different Pareto front, we see there are at least $r$ fronts $\F_i$ such that $\F_i \leqq x$.  Note also that for $j\leq i$,  $\F_i \leqq x \implies \F_j \leqq x$.  It follows that $\F_i \leqq x$ for $i=1,\dots,r$ and $u_n(x) = r \leq h_n(x)$.  For the opposite inequality, by definition of $h_n(x)$ there exists $x_k \in \F_k$ such that $x_k \leqq x$.  By the definition of $\F_k$, there exists $x_{k-1} \in \F_{k-1}$ such that $x_{k-1} \leqq x_k$.  By repeating this argument, we can find $x_1,\dots,x_k$ with $x_i \in \F_i$ and $x_1\leqq \cdots \leqq x_k$, hence we have exhibited a chain of length $k$ in $\X_n$ and it follows that $h_n(x) = k \leq u_n(x)$.  The proof is completed by noting that $X_1,\dots,X_n$ are distinct with probability one.
\end{proof}

 It is well-known \cite{ehrgott2005} that Pareto methods outperform more traditional linear scalarization methods when the Pareto fronts are non-convex.  In previous work~\cite{hsiao2012}, we showed that the Pareto fronts \emph{always} have microscopic non-convexities due to randomness, even when the Pareto fronts appear convex on a macroscopic scale.  Microscopic non-convexities only account for minor performance differences between Pareto methods and linear scalarization.  Macroscopic non-convexities induced by the geometry of the Pareto fronts on a macroscopic scale account for the major performance advantage of Pareto methods.  
 
 It is thus very important to characterize when the Pareto fronts are macroscopically convex.  We therefore make the following definition:
\begin{definition}\label{def:macroscopically-convex}
  Given a density $f:[0,1]^d \to [0,\infty)$, we say that $f$ yields \emph{macroscopically convex} Pareto fronts if for $X_1,\dots,X_n$ drawn \emph{i.i.d.}~from $f$ we have that the almost sure limit $U(x) :=\lim_{n\to \infty} n^{-\frac{1}{d}} h_n(x)$ exists for all $x$ and $U:[0,1]^d \to \R$ is quasiconcave.
\end{definition}
Recall that $U$ is said to be \emph{quasiconcave} if the super level sets 
\[\{x \in [0,1]^d \, : \, U(x) \geq a\}\]
are convex for all $a \in \R$. Since the Pareto fronts are encoded into the level sets of $h_n$, the asymptotic shape of the Pareto fronts is dictated by the level sets of the function $U$ from Definition \ref{def:macroscopically-convex}.  Hence the fronts are asymptotically convex on a macroscopic scale exactly when $U$ is quasiconcave, hence the definition.  

We now give our main result, which is a partial characterization of densities $f$ that yield macroscopically convex Pareto fronts.
\begin{theorem}\label{thm:macro}
  Let $f:[0,1]^d\to(0,\infty)$ be a continuous, log-concave, and separable density, i.e., $f(x)=f_1(x_1)\cdots f_d(x_d)$.  Then $f$ yields macroscopically convex Pareto fronts.
\end{theorem}
\begin{proof}
We denote by $F:[0,1]^d \to \R$ the cumulative distribution function (CDF) associated with the density $f$, which is defined by
\begin{equation}\label{eq:cdf}
F(x) = \int_0^{x_1} \cdots \int_0^{x_d} f(y_1,\dots,y_d) \, dy_1\cdots dy_d.
\end{equation}
Let $X_1,\dots,X_n$ be \emph{i.i.d.}~ with density $f$, and let $h_n$ denote the associated Pareto depth function, and $u_n$ the associated longest chain function.  By \cite[Theorem 1]{calder2014} we have that for every $x \in [0,1]^d$
\begin{equation}\label{eq:continuum}
  n^{-\frac{1}{d}} u_n(x) \longrightarrow U(x) \ \ \text{almost surely as } n \to \infty,
\end{equation}
where $U(x) = c_dF(x)^\frac{1}{d}$, and $c_d$ is a positive constant. In fact, the convergence is actually uniform on $[0,1]^d$ with probability one, but this is not necessary for the proof.  For a general non-separable density, the continuum limit \eqref{eq:continuum} still holds, but the limit $U(x)$ is not given by $c_dF(x)^\frac{1}{d}$ (it is instead the viscosity solution of a Hamilton-Jacobi equation), and the proof is quite involved (see~\cite{calder2014}). Fortunately, for the case of a separable density the proof is straightforward, and so we include it here for completeness.  

  Define $\Phi:[0,1]^d\to[0,1]^d$ by 
\[\Phi(x) = \left( \int_0^{x_1} f_1(t)\, dt , \dots, \int_0^{x_d} f_d(t) \, dt \right).\]
Since $f$ is continuous and strictly positive, $\Phi:[0,1]^d \to [0,1]^d$ is a $C^1$-diffeomorphism.  Setting $Y_i = \Phi(X_i)$, we easily see that $Y_1,\dots, Y_d$ are independent and uniformly distributed on $[0,1]^d$.  It is also easy to see that $\Phi$ preserves the partial order $\leqq$, i.e.,
\[x \leqq z \iff \Phi(x) \leqq \Phi(z).\]
Let $x \in [0,1]^d$, set $y=\Phi(x)$, and define $\Y_n = \Phi(\X_n)$.  By our above observations we have
\[u_n(x) = \max \{\ell \in \N \, : \, \exists \ y_1\leqq \cdots \leqq y_\ell \leqq y {\rm \ in \ } \Y_n\}.\]
Let $i_1 < \dots < i_N$ denote the indices of the random variables among $Y_1,\dots,Y_n$ that are less than or equal to $y$ and set $Z_k = Y_{i_k}$ for $k=1,\dots,N$.  Note that $N$ is binomially distributed with parameter $p:=F(x)$ and that $u_n(x)$ is the length of the longest chain among $N$ uniformly distributed points in the hypercube $\{z \in [0,1]^d \, : z \leqq y\}$.  By \cite[Remark 1]{bollobas1988} we have $N^{-\frac{1}{d}}u_n(x) \to c_d$ almost surely as $n\to \infty$ where $c_d < e$ are dimensional constants.  Since $n^{-1} N \to p$ almost surely as $n \to \infty$, we have
\[n^{-\frac{1}{d}} u_n(x) = \left(n^{-\frac{1}{d}} N^{\frac{1}{d}}\right) N^{-\frac{1}{d}} u_n(x)  \to c_d p^\frac{1}{d}\]
almost surely as $n\to \infty$.  The proof of \eqref{eq:continuum} is completed by recalling Proposition \ref{prop:chain}.

In the context of Definition \ref{def:macroscopically-convex}, we have $U(x) = c_dF(x)^\frac{1}{d}$. 
Hence $U$ is quasiconcave if and only if the cumulative distribution function $F$ is quasiconcave.  A sufficient condition for quasiconcavity of $F$ is log-concavity of $f$~\cite{prekopa1973}, which completes the proof.  
\end{proof}
Theorem \ref{thm:macro} indicates that Pareto methods are largely redundant when $f$ is a log-concave separable density.  As demonstrated in the {\bf Mediamill} \cite{snoek2006challenge} dataset (see Figure \ref{basicReal}), the distribution of points in Pareto space is not quasiconcave, and hence not log-concave, for the multiple-query retrieval problem.  This helps explain the success of our Pareto methods.  

 It would be very interesting to extend Theorem \ref{thm:macro} to arbitrary non-separable density functions $f$.  When $f$ is non-separable there is no simple integral expression like \eqref{eq:cdf} for $U$, and instead $U$ is characterized as the viscosity solution of a Hamilton-Jacobi partial differential equation~\cite[Theorem 1]{calder2014}.  This makes the non-separable case substantially more difficult, since $U$ is no longer an integral functional of $f$.   See~\cite{calder2014ITA} for a brief overview of our previous work on a continuum limit for non-dominated sorting~\cite{calder2013b,calder2014}.

\section{Multiple-query image retrieval}
\label{ParetoIR}

For most CBIR systems, images are preprocessed to extract low dimensional features instead of using pixel values directly for indexing and retrieval. Many feature extraction methods have been proposed in image processing and computer vision. In this paper we use the famous SIFT and HoG feature extraction techniques and apply spatial pyramid matching to obtain bag-of-words type features for image representation.  To avoid comparing every sample-query pair, we use an efficient manifold ranking algorithm proposed by \cite{xu2011}.

\subsection{Efficient manifold ranking (EMR)}
The traditional manifold ranking problem \cite{zhou2004ranking} is as follows. Let $\mathcal{X}=\{X_1,\ldots,X_n\} \subset \R^m$ be a finite set of points, and let $d:\mathcal{X} \times \mathcal{X} \rightarrow \R$ be a metric on $\mathcal{X}$, such as Euclidean distance. Define a vector $y = [y_1,\ldots,y_n]$, in which $y_i =1$ if $X_i$ is a query and $y_i = 0$ otherwise. Let $r:\mathcal{X} \rightarrow \R$ denote the ranking function which assigns a ranking score $r_i$ to each point $X_i$. The query is assigned a rank of $1$ and all other samples will be assigned smaller ranks based on their distance to the query along the manifold underlying the data.  To construct a graph on $\mathcal{X}$, first sort the pairwise distances between all samples in ascending order, and then add edges between points according to this order until a connected graph $G$ is constructed. The edge weight between $X_i$ and $X_j$ on this graph is denoted by $w_{ij}$. If there is an edge between $X_i$ and $X_j$, define the weight by $w_{ij} = exp[-d^2(X_i,X_j)/2\sigma^2]$, and if not, set $w_{ij}=0$, and set $W = (w_{ij})_{ij} \in \R^{n\times n}$.  In the manifold ranking method, the cost function associated with ranking vector $r$ is defined by 
\begin{equation*}
O(r)= \sum^{n}_{i,j = 1}w_{ij}|\frac{1}{\sqrt{D_{ii}}}r_i -\frac{1}{\sqrt{D_{jj}}}r_j|^2 + \mu \sum^{n}_{i=1}|r_i-y_i|^2\\
\end{equation*}
where $D$ is a diagonal matrix with $D_{ii} = \sum^{n}_{j=1}w_{ij}$ and $\mu >0$ is the regularization parameter. The first term in the cost function is a smoothness term that forces nearby points have similar ranking scores.  The second term is a regularization term, which forces the query to have a rank close to $1$, and all other samples to have ranks as close to $0$ as possible.  The ranking function $r$ is the minimizer of $O(r)$ over all possible ranking functions.  

This optimization problem can be solved in either of two ways: a direct approach and an iterative approach. The direct approach computes the exact solution via the closed form expression
\begin{equation}\label{rstar}
r^* = (I_n-\alpha S)^{-1}y 
\end{equation}
where $\alpha = \frac{1}{1+\mu}$, $I_n$ is an $n\times n$ identity matrix and $S = D^{-1/2}WD^{-1/2}$. The iterative method is better suited to large scale datasets. The ranking function $r$ is computed by repeating the iteration scheme$r(t+1) = \alpha Sr(t)+(1-\alpha)y,$
until convergence.
The direct approach requires an $n\times n$ matrix inversion and the iterative approach requires $n\times n $  memory and may converge to a local minimum. In addition, the complexity of constructing the graph $G$ is $O(n^2\log n)$.  Sometimes a $kNN$ graph is used for $G$, in which case the complexity is $O(kn^2)$. Neither case is suitable for large-scale problems.

In \cite{xu2011}, an efficient manifold ranking algorithm is proposed. The authors introduce an anchor graph $U$ to model the data and use the Nadaraya-Watson kernel to construct a weight matrix $Z\in \R^{d\times n}$ which measures the potential relationships between data points in $\mathcal{X}$ and anchors in $U$. For convenience, denote by $z_i$ the $i$-th column of $Z$. The affinity matrix $W$ is then designed to be $Z^TZ$. The final ranking function $r$ can then be directly computed by 
\begin{equation}\label{rstar2}
r^* = (I_n-H^T(HH^T-\frac{1}{\alpha}I_d)^{-1}H)y,
\end{equation}
where $H = ZD^{-\frac{1}{2}}$ and $D$ is a diagonal matrix with $D_{ii} = \sum^{n}_{j=1}z^T_iz_j$.
This method requires inverting only a $d\times d$ matrix, in contrast to inverting the $n\times n$ matrix used in standard manifold ranking.
When $d\ll n$, as occurs in large databases, the computational cost of manifold ranking is significantly reduced. The complexity of computing the ranking function with the EMR algorithm is $O(dn+d^3)$. In addition, EMR does not require storage of an $n\times n$ matrix.

Notice construction of the anchor graph and computation of the matrix inversion \cite{xu2011} can be implemented offline. For out-of-sample retrieval, Xu et al.~\cite{xu2011} provides an efficient way to update the graph structure and do out-of-sample retrieval quickly.

\setcounter{subfigure}{0}

\subsection{Multiple-query case}

In \cite{xu2011}, prior knowledge about the relevance or confidence of each query can be incorporated into the EMR algorithm through the choice of the initial vector $y$. For example, in the multiple-query information retrieval problem, we may have queried, say, $X_1,X_2$ and $X_3$. We could set $y_1=y_2=y_3=1$ and $y_i=0$ for $i\geq 4$ in the EMR algorithm.  This instructs the EMR algorithm to assign high ranks to $X_1$, $X_2$, and $X_3$ in the final ranking $r^*$. It is easy to see from \eqref{rstar} or \eqref{rstar2} that $r^*$ is equal to the scalarization $r^*_1+r^*_2+r^*_3$ where $r^*_i$, $i=1,2,3,$ is the ranking function obtained when issuing each query individually. The main contribution of this paper is to show that our Pareto front method can outperform this standard linear scalarization method. Our proposed algorithm is given below.

Given a set of queries $\{q_1, q_2,...,q_T\}$, we apply EMR to compute the ranking functions $r^*_1\ldots r^*_T \in \R^N$ corresponding to each query. We then define the dissimilarity vector $d_i\in\R^N_+$ between $q_i$ and all samples by $d_i = \textbf{1}-r^*_i$ where $\textbf{1} = [1,\dots,1]\in \R^N$.  We then construct the Pareto fronts associated to $d_1,\dots,d_T$ as described in Section \ref{IRusingPRM}. To return relevant samples to the user, we return samples according to their Pareto front number, i.e., we return points on $\F_1$ first, then $\F_2$, and so on, until sufficiently many images have been retrieved. Within the same front, we return points in a particular order, e.g., for $T=2$, from the middle first. In the context of this paper, the middle part of each front will contain  samples related to \emph{all} queries. Relevance feedback schemes can also be used with our algorithm to enhance retrieval performance. For example one could use images labeled as relevant by the user as new queries to generate new Pareto fronts.

\section{Experimental study}
\label{experiment}

We now present an experimental study comparing our Pareto front method against several state of the art multiple-query retrieval algorithms.  Since our proposed algorithm was developed for the case where each query corresponds to a different semantic, we use multi-label datasets in our experiments.  By multi-label, we mean that many images in the dataset belong to more than one class.  This allows us to measure in a precise way our algorithm's ability to find images that are similar to all queries.   

\subsection{Multiple-query performance metrics}

We evaluate the performance of our algorithm using normalized Discounted Cumulative Gain (nDCG) \cite{jarvelin2002cumulated}, which is standard in the retrieval community.  The nDCG is defined in terms of a relevance score, which measures the relevance of a returned image to the query. In single query retrieval, a popular relevance score is the binary score, which is 1 if the retrieved image is related to the query and 0 otherwise.  In the context of multiple-query retrieval, where a retrieved image may be related to each query in distinct ways, the binary score is an oversimplification of the notion of relevance. Therefore, we define a new relevance score for performance assessment of multiple-query multiclass retrieval algorithms.  We call this relevance score  multiple-query unique relevance (mq-uniq-rel).

Roughly speaking, multiple-query unique relevance measures the fraction of query class labels that are covered by the retrieved object when the retrieved object is uniquely related to each query.  When the retrieved object is not uniquely related to each query, the relevance score is set to zero.  The importance of having unique relations to each query cannot be understated.  For instance, in the two-query problem, if a retrieved image is related to one of the queries only through a feature common to both queries, then the image is effectively relevant to only one of the those queries in the sense it would likely be ranked highly by a single-query retrieval algorithm issuing only one of the queries.  A more interesting and challenging problem, which is the focus of this paper, is to find images that have different features in common with each query.

More formally, let us denote by $C$ the total number of classes in the dataset, and let $\ell \in \{0,1\}^C$ be the binary label vector of a retrieved object $X$.  Similarly, let $y^i$ be the label vector of query $q_i$.  Given two label vectors $\ell^1$ and $\ell^2$, we denote by the logical disjunction $\ell^1\vee \ell^2$ (respectively, the logical conjunction $\ell^1\wedge \ell^2$) the label vector whose $j^{\rm th}$ entry is given by $\max(\ell^1_j,\ell^2_j)$ (respectively, $\min(\ell^1_j,\ell^2_j)$).   We denote by $|\ell|$ the number of non-zero entries in the label vector $\ell$.   Given a set of queries $\{q_1,\ldots,q_T\}$, we define the multiple-query unique relevance (mq-uniq-rel) of retrieved sample $X$ having label $\ell$ to the query set by
\begin{equation}\label{eq:mp2}
\text{mq-uniq-rel}(X) = \begin{cases}
\displaystyle \frac{|\ell \wedge \beta|}{|\beta|}, & \text{if } \forall i, |\ell \wedge (y^i-\eta^i)| \neq 0,\\
0,&  {\rm otherwise},
\end{cases}
\end{equation}
where $\beta = y^1\vee y^2 \vee\cdots \vee y^T$ is the disjunction  of the label vectors corresponding to $q_1,\dots,q_T$ and $\eta^i = \bigvee_{j\neq i} y^j\wedge y^i$.  Multiple-query unique relevance measures the fraction of query classes that the retrieved object belongs to whenever the retrieved image has a unique relation to each query, and is set to zero otherwise.

For simplicity of notation, we denote by $\text{mq-uniq-rel}_i$ the multiple-query unique relevance of the $i^{\rm th}$ retrieved image. The  Discounted Cumulative Gain (DCG) is then given by 
\begin{equation}\label{eq:mq-DCG}
\text{DCG} = \text{mq-uniq-rel}_1 + \sum_{i=2}^k \frac{\text{mq-uniq-rel}^1_i}{\log_2(i)},
\end{equation} 
The normalized DCG, or nDCG, is computed by normalizing the DCG by the best possible score which is $ 1 + \sum_{i=1}^k \frac{1}{\log_2(i)}$.

Note that, analogous to binary relevance score, we have $\text{mq-uniq-rel}_i = 1$ if and only if the label vector corresponding to the $i^{\rm th}$ retrieved object contains all labels from both queries and each query has at least one unique class.  The difference is that multiple-query relevance is not a binary score, and instead assigns a range of values between zero and one, depending on how many of the query labels are covered by the retrieved object. Thus, $\text{mq-uniq-rel}_i$ can be viewed as a generalization of the binary relevance score to the multiple-query setting in which the goal is to find objects uniquely related to all queries.

\subsection{Evaluation on multi-label datasets}
\begin{figure*}[t!]
\vskip 0.2in
\centering
\subfigure[{\bf Mediamill} dataset]{\includegraphics[width=0.9\columnwidth]{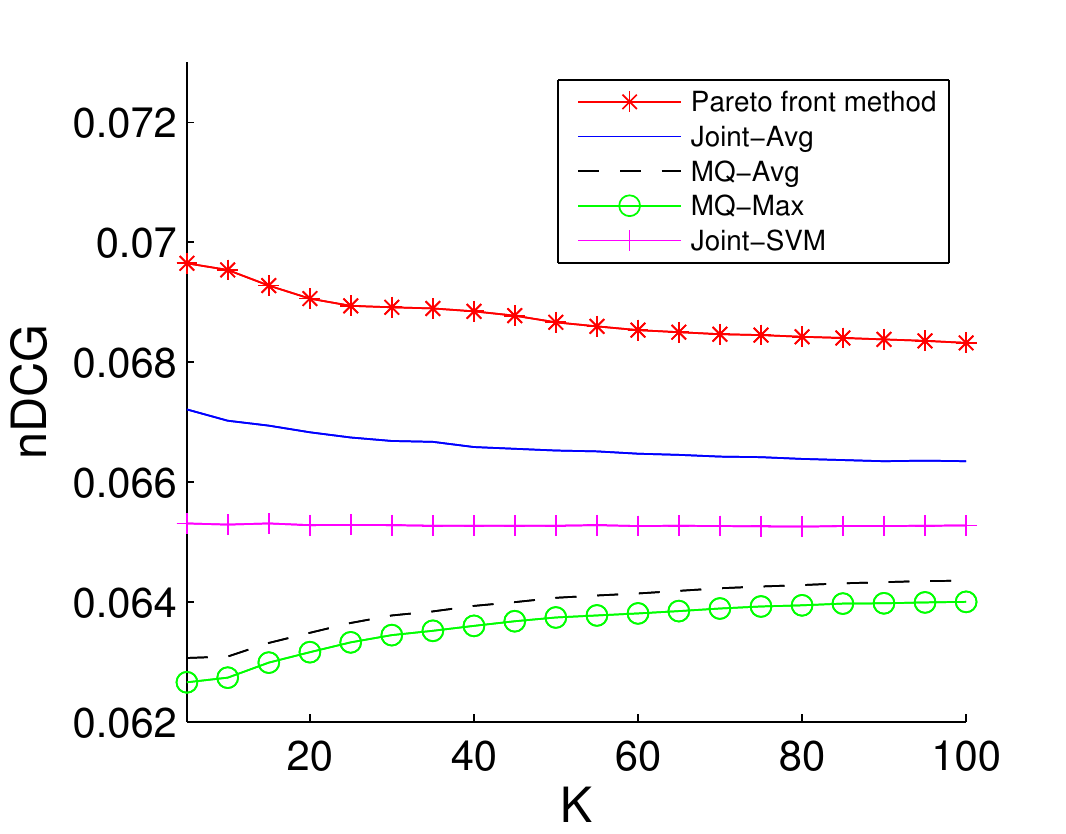}\label{fig:ndcg_large_withrules}}
\subfigure[{\bf LAMDA} dataset]{\includegraphics[width=0.9\columnwidth]{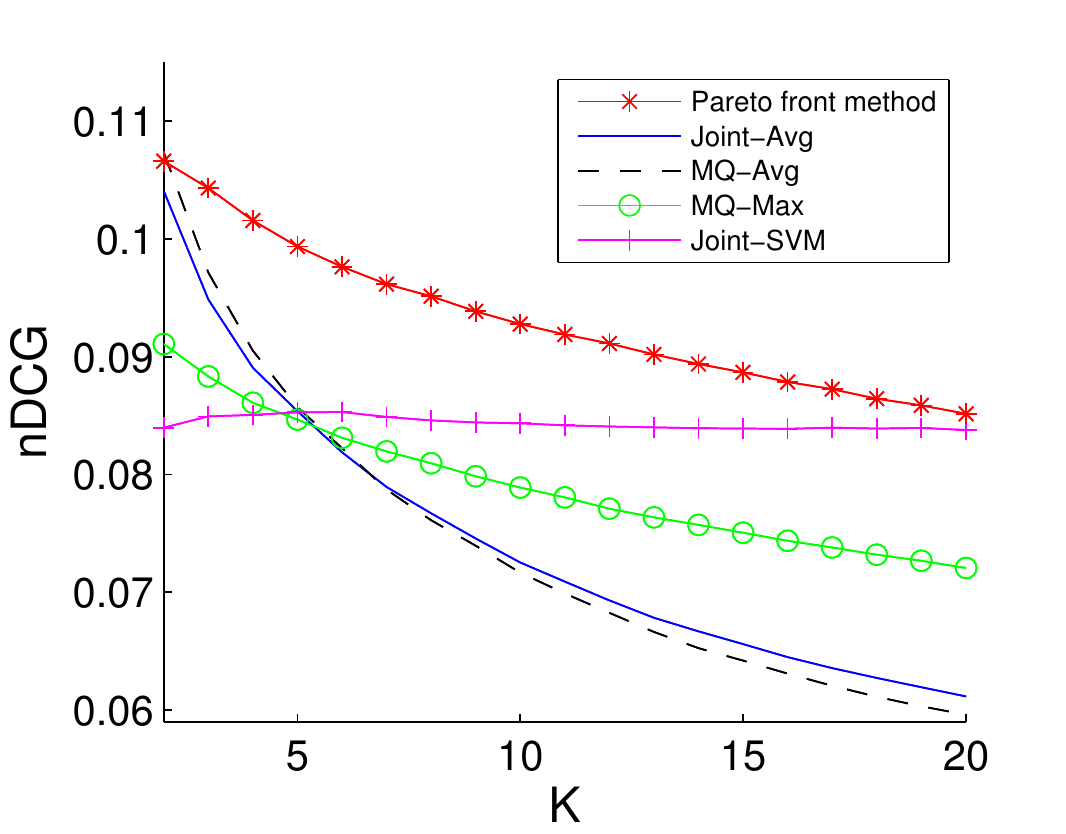}\label{fig:ndcg_small_withrules}}
\caption{Comparison of PFM against state of the art multiple-query retrieval algorithms for {\bf LAMDA} and {\bf Mediamill} dataset with respect to the nDCG defined in (\ref{eq:mq-DCG}). The proposed method significantly outperforms others on both datasets.
}
\end{figure*} 

We evaluate our algorithm on the {\bf Mediamill} video dataset \cite{snoek2006challenge}, which has been widely used for benchmarking multi-class classifiers, and the {\bf LAMDA} dataset, which is widely used in the retrieval community~\cite{zhou2006}. The {\bf Mediamill} dataset consists of 29800 videos, and Snoek et al.~\cite{snoek2006challenge} provide  a manually annotated lexicon containing 101 semantic concepts and a large number of pre-computed low-level multimedia features. Each visual feature vector, $X_i$, is a 120-dimensional vector that corresponds to a specified key frame in the dataset. The feature extraction is based on the method of \cite{van2006robust} and characterizes both global and local color-texture information of a single key frame, that is, an image. Each key frame is associated with a label vector $\ell \in \{0,1\}^C$, and each entry of $\ell$ corresponds to one of 101 semantic concepts. If $X_i$ contains content related to the $j^{\rm th}$ semantic concept, then the $j^{\rm th}$ entry of $\ell_i$ is set to 1, and if not, it is set to 0. 

The {\bf LAMDA} database contains 2000 images, and each image has been labeled with one or more of the following five class labels: desert, mountains, sea, sunset, and trees.  In total, 1543 images belong to exactly one class, 442 images belong to exactly two classes, and 15 images belong to three classes.  Of the 442 two-class images, 106 images have the labels `mountain' and `sky', 172 images have the labels `sunset' and `sea', and the remaining label-pairs each have less than 40 image members, with some as few as 5. Zhou and Zhang \cite{zhou2006} preprocessed the database and extracted from  each image a 135 element feature vector, which we use to compute image similarities.  

To evaluate the performance of our algorithm, we randomly generated 10000 query-pairs for {\bf Mediamill} and 1000 for {\bf LAMDA}, and ran our multiple-query retrieval algorithm on each pair.  We computed the $\text{nDCG}$ for different retrieval algorithms for each query-pair, and then computed the average nDCG over all query-pairs at different $K$. Since Efficient Manifold Ranking (EMR) uses a random initial condition for constructing the anchor graph, we furthermore run the entire experiment 20 times for {\bf Mediamill} and 100 times for {\bf LAMDA}, and computed the mean nDCG over all experiments.  This ensures that we avoid any bias from a particular EMR model.   

We show the mean $\text{nDCG}$ for our algorithm and many state of the art multiple-query retrieval algorithms for {\bf Mediamill} and {\bf LAMDA} in Figures \ref{fig:ndcg_large_withrules} and \ref{fig:ndcg_small_withrules}, respectively. We compare against MQ-Avg, MQ-Max, Joint-Avg, and Joint-SVM~\cite{arandjelovic2012multiple}. Joint-Avg combines histogram features of different queries to generate a new feature vector to use as an out-of-example query. A Joint-SVM classifier is used to rank each sample in response to each query. We note that Joint-SVM does not use EMR,  while MQ-Avg and MQ-Max both do.  Figures \ref{fig:ndcg_large_withrules} and \ref{fig:ndcg_small_withrules} show that our retrieval algorithm significantly outperforms all other algorithms.  

We should note that when randomly generating query-pairs for {\bf LAMDA}, we consider only the label-pairs (`mountain',`sky') and (`sunset',`sea'), since these are the only label-pairs for which there are a significant number of corresponding two-class images.  If there no multi-class images corresponding to a query-pair, then multiple-query retrieval is unnecessary; one can simply issue each query separately and take a union of the retrieved images. 

\begin{figure}[t]
\includegraphics[width=8.3cm]{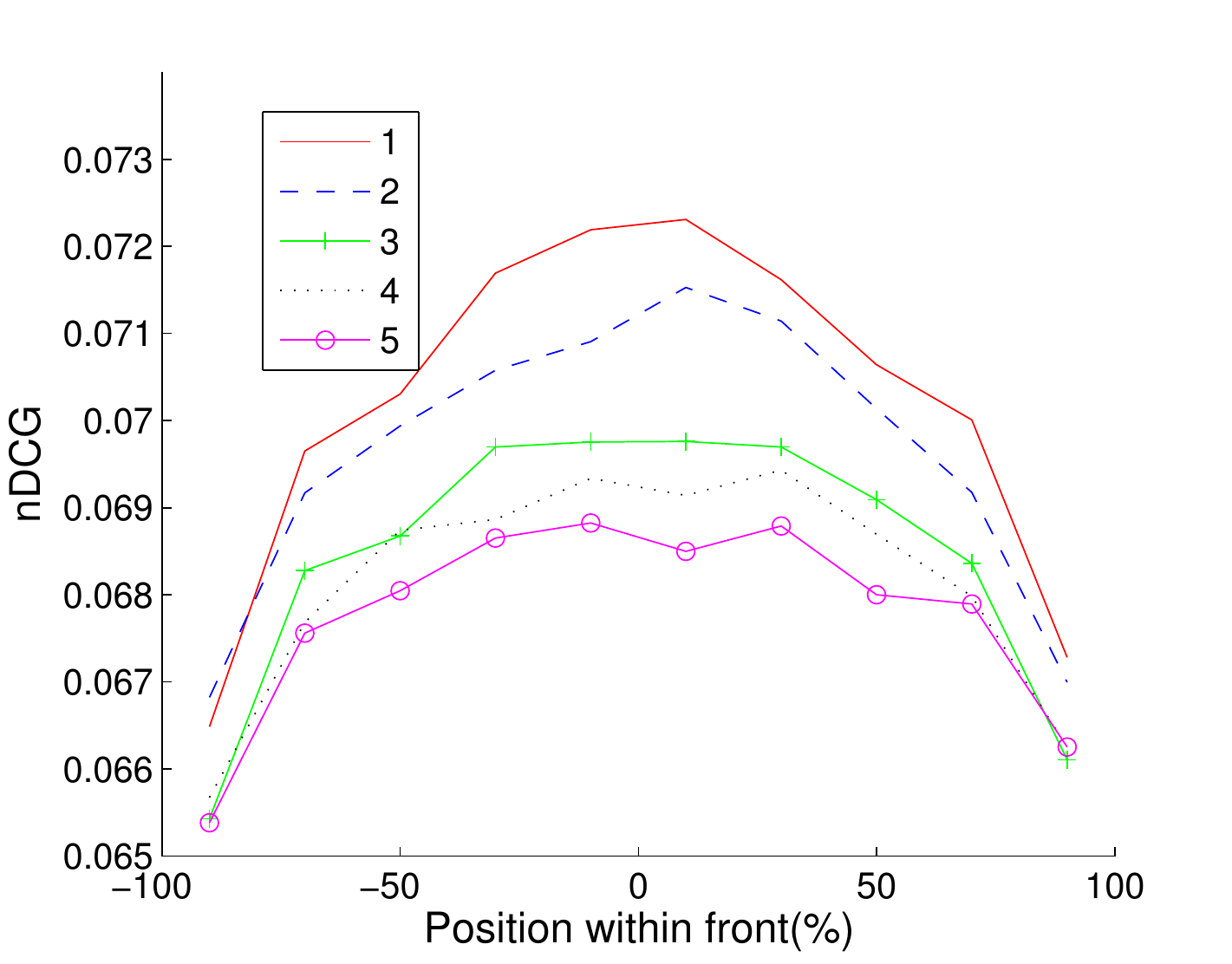}
\caption{Average unique relevance scores at different regions along top five Pareto fronts. This plot validates our assumption that the middle part of first Pareto fronts contain more important samples that are uniquely related to both queries. Samples at deeper fronts and near the tails are less interesting.}
\label{PrecisionPlotLinearFig}
\end{figure}

\begin{figure*}[ht]
\vskip 0.2in
\begin{center}
\centerline{\includegraphics[width=14 cm]{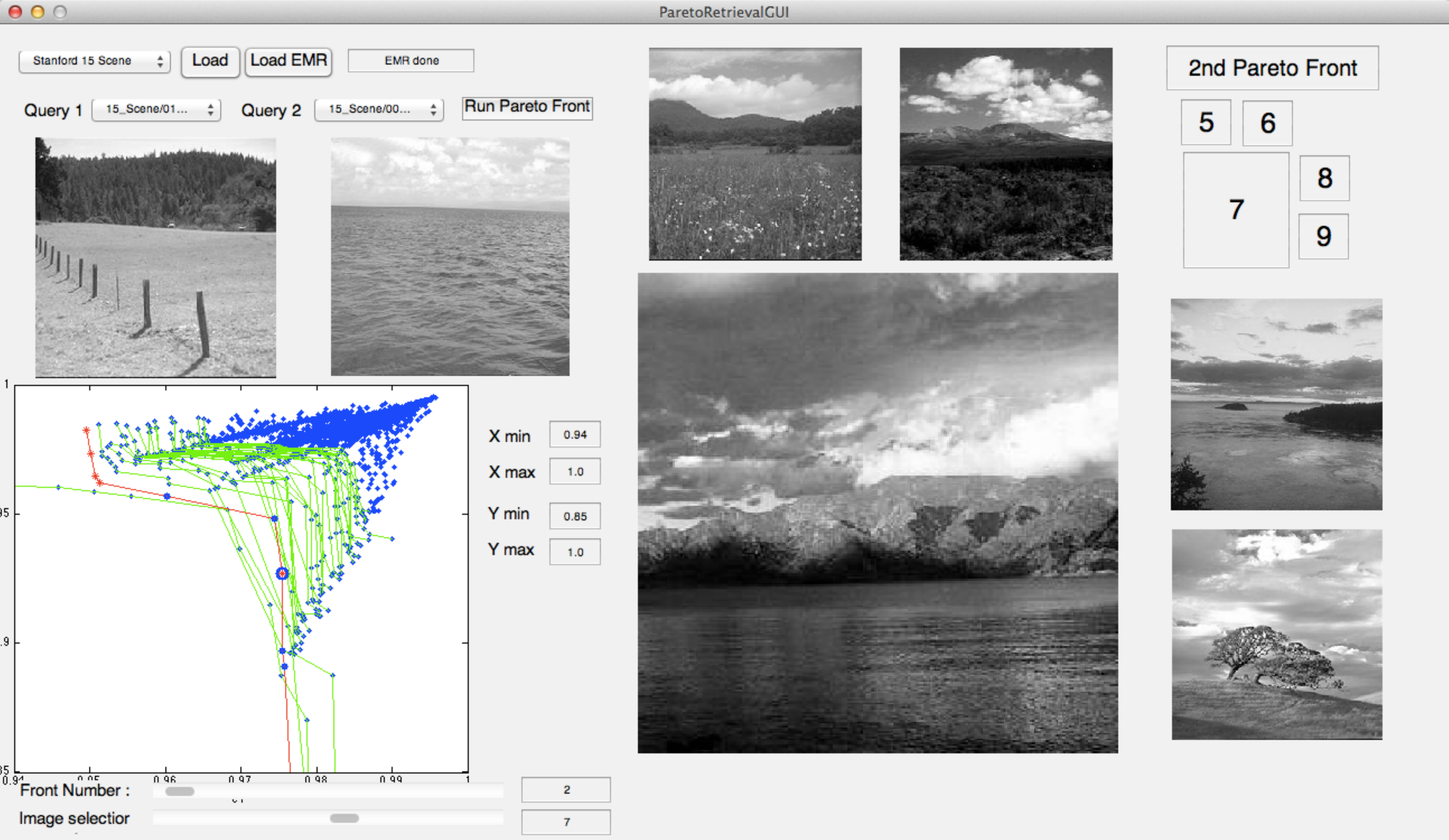}}
\caption{GUI screenshot. The two images on the upper left are two query images containing mountain and water, respectively. The largest image corresponds to the $7^{th}$ Pareto point on the second Pareto front and the other four images correspond to adjacent points on the same front . Users can select the front and the specific relevance point using the two slider bars at the bottom. }
\label{GUIscreenshot}
\end{center}
\vskip -0.2in
\end{figure*} 
  
To visualize advantages of the Pareto front method, we show in Figure  \ref{PrecisionPlotLinearFig} the multiple-query unique relevance scores for points within each of the first five Pareto fronts, plotted from one tail of the front, to the middle, to the other tail.  The relevance scores within each front are interpolated to a fixed grid, and averaged over all query pairs to give the curves in Figure \ref{PrecisionPlotLinearFig}.  We used the \textbf{Mediamill} dataset to generate Figure \ref{PrecisionPlotLinearFig}; the results on \textbf{LAMDA} are similar.
This result  validates  our assumption that the first front includes more important samples than deeper fronts, and that the middle of the Pareto fronts is fundamentally important for multiple-query retrieval. 

We also note that the middle portions of fronts 2--5 contain samples with higher scores than those located at the tail of the first front. This phenomenon suggests a modified version of PFM which starts returning points around the middle of second front after returning only, say, $d$ points from the first front.  The same would hold for the second front and so on. We have carried out some experiments with such an algorithm, and have found that it can lead to even larger performance improvements, as suggested by Figure \ref{PrecisionPlotLinearFig}, for certain choices of $d$. However, it may be difficult to determine the best choice of $d$ in advance since the label information is not available. Recall that label information is available only for testing and generating Figure \ref{PrecisionPlotLinearFig} for validation. Therefore, we decided for simplicity to leave this simple modification of the algorithm to future work.

\subsection{GUI for Pareto front retrieval}
A GUI for a two-query image retrieval was implemented to illustrate the Pareto front method for image retrieval. Users can easily select samples from different fronts and visually explore the neighboring samples along the front. Samples corresponding to Pareto points at one tail of the front are similar to only one query, while samples corresponding to Pareto points at the middle part of front are similar to both queries. When the Pareto point cloud is non-convex, users can use our GUI to easily identify  the Pareto points that cannot be obtained by any linear scalarization method. The screen shot of our GUI is shown in Figure \ref{GUIscreenshot}. In this example, the two query images correspond to a mountain and an ocean respectively. One of the retrieved images corresponds to a point in the middle part of the second front that includes both a mountain and an ocean. The Matlab code of GUI can be downloaded from \url{http://tbayes.eecs.umich.edu/coolmark/pareto}.

\section{Conclusions}
We have presented a novel algorithm for content-based multiple-query image retrieval where the queries all correspond to different image semantics, and the goal is to find images related to \emph{all} queries. This algorithm can retrieve samples which are not easily retrieved by other multiple-query retrieval algorithms and any linear scalarization method. We have presented theoretical results on asymptotic non-convexity of Pareto fronts that proves that the Pareto approach is better than using linear combinations of ranking results. Experimental studies on real-world datasets illustrate the advantages of the proposed Pareto front method.

\bibliographystyle{IEEEtran}
\bibliography{ParetoRetrievalRef}

\begin{thebibliography}{10}
\providecommand{\url}[1]{#1}
\csname url@samestyle\endcsname
\providecommand{\newblock}{\relax}
\providecommand{\bibinfo}[2]{#2}
\providecommand{\BIBentrySTDinterwordspacing}{\spaceskip=0pt\relax}
\providecommand{\BIBentryALTinterwordstretchfactor}{4}
\providecommand{\BIBentryALTinterwordspacing}{\spaceskip=\fontdimen2\font plus
\BIBentryALTinterwordstretchfactor\fontdimen3\font minus
  \fontdimen4\font\relax}
\providecommand{\BIBforeignlanguage}[2]{{%
\expandafter\ifx\csname l@#1\endcsname\relax
\typeout{** WARNING: IEEEtran.bst: No hyphenation pattern has been}%
\typeout{** loaded for the language `#1'. Using the pattern for}%
\typeout{** the default language instead.}%
\else
\language=\csname l@#1\endcsname
\fi
#2}}
\providecommand{\BIBdecl}{\relax}
\BIBdecl

\bibitem{frome2007image}
A.~Frome, Y.~Singer, and J.~Malik, ``Image retrieval and classification using
  local distance functions,'' in \emph{Advances in Neural Information
  Processing Systems: Proceedings of the 2006 Conference}, vol.~19.\hskip 1em
  plus 0.5em minus 0.4em\relax The MIT Press, 2007, p. 417.

\bibitem{gia2004instance}
G.~Gia, F.~Roli \emph{et~al.}, ``Instance-based relevance feedback for image
  retrieval,'' in \emph{Advances in neural information processing systems},
  2004, pp. 489--496.

\bibitem{vasconcelos1999learning}
N.~Vasconcelos and A.~Lippman, ``Learning from user feedback in image retrieval
  systems.'' in \emph{NIPS}, 1999, pp. 977--986.

\bibitem{belkin1993effect}
N.~J. Belkin, C.~Cool, W.~B. Croft, and J.~P. Callan, ``The effect of multiple
  query representations on information retrieval system performance,'' in
  \emph{Proceedings of the 16th annual international ACM SIGIR conference on
  Research and development in information retrieval}.\hskip 1em plus 0.5em
  minus 0.4em\relax ACM, 1993, pp. 339--346.

\bibitem{arandjelovic2012multiple}
R.~Arandjelovic and A.~Zisserman, ``Multiple queries for large scale specific
  object retrieval,'' in \emph{British Machine Vision Conference}, 2012.

\bibitem{jin2005improving}
X.~Jin and J.~French, ``Improving image retrieval effectiveness via multiple
  queries,'' \emph{Multimedia Tools and Applications}, vol.~26, no.~2, pp.
  221--245, 2005.

\bibitem{xu2011}
B.~Xu, J.~Bu, C.~Chen, D.~Cai, X.~He, W.~Liu, and J.~Luo, ``Efficient manifold
  ranking for image retrieval,'' in \emph{Proceedings of the 34th international
  ACM SIGIR conference on Research and development in Information}.\hskip 1em
  plus 0.5em minus 0.4em\relax ACM, 2011, pp. 525--534.

\bibitem{ehrgott2005}
M.~Ehrgott, \emph{{Multicriteria Optimization (2. ed.)}}.\hskip 1em plus 0.5em
  minus 0.4em\relax Springer, 2005.

\bibitem{jeon2003automatic}
J.~Jeon, V.~Lavrenko, and R.~Manmatha, ``Automatic image annotation and
  retrieval using cross-media relevance models,'' in \emph{Proceedings of the
  26th annual international ACM SIGIR conference on Research and development in
  informaion retrieval}.\hskip 1em plus 0.5em minus 0.4em\relax ACM, 2003, pp.
  119--126.

\bibitem{carneiro2007supervised}
G.~Carneiro, A.~B. Chan, P.~J. Moreno, and N.~Vasconcelos, ``Supervised
  learning of semantic classes for image annotation and retrieval,''
  \emph{Pattern Analysis and Machine Intelligence, IEEE Transactions on},
  vol.~29, no.~3, pp. 394--410, 2007.

\bibitem{russell2008labelme}
B.~C. Russell, A.~Torralba, K.~P. Murphy, and W.~T. Freeman, ``Labelme: a
  database and web-based tool for image annotation,'' \emph{International
  journal of computer vision}, vol.~77, no. 1-3, pp. 157--173, 2008.

\bibitem{liu2007survey}
Y.~Liu, D.~Zhang, G.~Lu, and W.~Ma, ``A survey of content-based image retrieval
  with high-level semantics,'' \emph{Pattern Recognition}, vol.~40, no.~1, pp.
  262--282, 2007.

\bibitem{datta2008image}
R.~Datta, D.~Joshi, J.~Li, and J.~Wang, ``Image retrieval: Ideas, influences,
  and trends of the new age,'' \emph{ACM Computing Surveys (CSUR)}, vol.~40,
  no.~2, p.~5, 2008.

\bibitem{zloof1975query}
M.~Zloof, ``Query by example,'' in \emph{Proceedings of the May 19-22, 1975,
  national computer conference and exposition}.\hskip 1em plus 0.5em minus
  0.4em\relax ACM, 1975, pp. 431--438.

\bibitem{hirata1992query}
K.~Hirata and T.~Kato, ``Query by visual example,'' in \emph{Advances in
  Database Technology—EDBT'92}.\hskip 1em plus 0.5em minus 0.4em\relax
  Springer, 1992, pp. 56--71.

\bibitem{lowe2004}
D.~Lowe, ``Distinctive image features from scale-invariant keypoints,''
  \emph{International journal of computer vision}, vol.~60, no.~2, pp. 91--110,
  2004.

\bibitem{dalal2005histograms}
N.~Dalal and B.~Triggs, ``Histograms of oriented gradients for human
  detection,'' in \emph{Computer Vision and Pattern Recognition, 2005. CVPR
  2005. IEEE Computer Society Conference on}, vol.~1.\hskip 1em plus 0.5em
  minus 0.4em\relax IEEE, 2005, pp. 886--893.

\bibitem{liu2009learning}
T.~Liu, ``Learning to rank for information retrieval,'' \emph{Foundations and
  Trends in Information Retrieval}, vol.~3, no.~3, pp. 225--331, 2009.

\bibitem{burges2005learning}
C.~Burges, T.~Shaked, E.~Renshaw, A.~Lazier, M.~Deeds, N.~Hamilton, and
  G.~Hullender, ``Learning to rank using gradient descent,'' in
  \emph{Proceedings of the 22nd international conference on Machine
  learning}.\hskip 1em plus 0.5em minus 0.4em\relax ACM, 2005, pp. 89--96.

\bibitem{brin1998anatomy}
S.~Brin and L.~Page, ``The anatomy of a large-scale hypertextual web search
  engine,'' \emph{Computer networks and ISDN systems}, vol.~30, no.~1, pp.
  107--117, 1998.

\bibitem{zhou2004ranking}
D.~Zhou, J.~Weston, A.~Gretton, O.~Bousquet, and B.~Sch{\"o}lkopf, ``Ranking on
  data manifolds,'' \emph{Advances in neural information processing systems},
  vol.~16, pp. 169--176, 2004.

\bibitem{zhou2004learning}
D.~Zhou, O.~Bousquet, T.~Lal, J.~Weston, and B.~Sch{\"o}lkopf, ``Learning with
  local and global consistency,'' \emph{Advances in neural information
  processing systems}, vol.~16, pp. 321--328, 2004.

\bibitem{jin2008pareto}
Y.~Jin and B.~Sendhoff, ``Pareto-based multiobjective machine learning: An
  overview and case studies,'' \emph{Systems, Man, and Cybernetics, Part C:
  Applications and Reviews, IEEE Transactions on}, vol.~38, no.~3, pp.
  397--415, 2008.

\bibitem{borzsony2001skyline}
S.~Borzsony, D.~Kossmann, and K.~Stocker, ``The skyline operator,'' in
  \emph{Data Engineering, 2001. Proceedings. 17th International Conference
  on}.\hskip 1em plus 0.5em minus 0.4em\relax IEEE, 2001, pp. 421--430.

\bibitem{kossmann2002shooting}
D.~Kossmann, F.~Ramsak, and S.~Rost, ``Shooting stars in the sky: an online
  algorithm for skyline queries,'' in \emph{Proceedings of the 28th
  international conference on Very Large Data Bases}.\hskip 1em plus 0.5em
  minus 0.4em\relax VLDB Endowment, 2002, pp. 275--286.

\bibitem{tan2001efficient}
K.-L. Tan, P.-K. Eng, B.~C. Ooi \emph{et~al.}, ``Efficient progressive skyline
  computation,'' in \emph{Proceedings of the International Conference on Very
  Large Data Bases}, 2001, pp. 301--310.

\bibitem{papadias2003optimal}
D.~Papadias, Y.~Tao, G.~Fu, and B.~Seeger, ``An optimal and progressive
  algorithm for skyline queries,'' in \emph{Proceedings of the 2003 ACM SIGMOD
  international conference on Management of data}.\hskip 1em plus 0.5em minus
  0.4em\relax ACM, 2003, pp. 467--478.

\bibitem{hristidis2001prefer}
V.~Hristidis, N.~Koudas, and Y.~Papakonstantinou, ``Prefer: A system for the
  efficient execution of multi-parametric ranked queries,'' \emph{ACM SIGMOD
  Record}, vol.~30, no.~2, pp. 259--270, 2001.

\bibitem{jin2004mining}
W.~Jin, J.~Han, and M.~Ester, ``Mining thick skylines over large databases,''
  in \emph{Knowledge Discovery in Databases: PKDD 2004}.\hskip 1em plus 0.5em
  minus 0.4em\relax Springer, 2004, pp. 255--266.

\bibitem{agrawal2000framework}
R.~Agrawal and E.~L. Wimmers, ``A framework for expressing and combining
  preferences,'' in \emph{ACM SIGMOD Record}, vol.~29, no.~2.\hskip 1em plus
  0.5em minus 0.4em\relax ACM, 2000, pp. 297--306.

\bibitem{jensen2003}
M.~T. Jensen, ``Reducing the run-time complexity of multiobjective {EA}s: The
  {NSGA-II} and other algorithms,'' \emph{{IEEE} Trans. Evol. Comput.}, vol.~7,
  no.~5, pp. 503--515, 2003.

\bibitem{sharifzadeh2006spatial}
M.~Sharifzadeh and C.~Shahabi, ``The spatial skyline queries,'' in
  \emph{Proceedings of the 32nd international conference on Very large data
  bases}.\hskip 1em plus 0.5em minus 0.4em\relax VLDB Endowment, 2006, pp.
  751--762.

\bibitem{hero2004pareto}
A.~Hero and G.~Fleury, ``Pareto-optimal methods for gene ranking,'' \emph{The
  Journal of VLSI Signal Processing}, vol.~38, no.~3, pp. 259--275, 2004.

\bibitem{hsiao2012}
K.-J. Hsiao, K.~Xu, J.~Calder, and A.~Hero, ``{Multi-criteria anomaly detection
  using Pareto Depth Analysis},'' \emph{{Advances in Neural Information
  Processing Systems}}, vol.~25, pp. 854--862, 2012.

\bibitem{aslam2001models}
J.~A. Aslam and M.~Montague, ``Models for metasearch,'' in \emph{Proceedings of
  the 24th annual international ACM SIGIR conference on Research and
  development in information retrieval}.\hskip 1em plus 0.5em minus 0.4em\relax
  ACM, 2001, pp. 276--284.

\bibitem{meng2002building}
W.~Meng, C.~Yu, and K.-L. Liu, ``Building efficient and effective metasearch
  engines,'' \emph{ACM Computing Surveys (CSUR)}, vol.~34, no.~1, pp. 48--89,
  2002.

\bibitem{lee1997analyses}
J.~H. Lee, ``Analyses of multiple evidence combination,'' in \emph{ACM SIGIR
  Forum}, vol.~31, no.~SI.\hskip 1em plus 0.5em minus 0.4em\relax ACM, 1997,
  pp. 267--276.

\bibitem{blum1998combining}
A.~Blum and T.~Mitchell, ``Combining labeled and unlabeled data with
  co-training,'' in \emph{Proceedings of the eleventh annual conference on
  Computational learning theory}.\hskip 1em plus 0.5em minus 0.4em\relax ACM,
  1998, pp. 92--100.

\bibitem{sindhwani2005co}
V.~Sindhwani, P.~Niyogi, and M.~Belkin, ``A co-regularization approach to
  semi-supervised learning with multiple views,'' in \emph{Proceedings of ICML
  Workshop on Learning with Multiple Views}.\hskip 1em plus 0.5em minus
  0.4em\relax Citeseer, 2005, pp. 74--79.

\bibitem{christoudias2012multi}
C.~Christoudias, R.~Urtasun, and T.~Darrell, ``Multi-view learning in the
  presence of view disagreement,'' \emph{arXiv preprint arXiv:1206.3242}, 2012.

\bibitem{gonen2011multiple}
M.~G{\"o}nen and E.~Alpayd{\i}n, ``Multiple kernel learning algorithms,''
  \emph{Journal of Machine Learning Research}, vol.~12, pp. 2211--2268, 2011.

\bibitem{calder2014}
J.~Calder, S.~Esedoglu, and A.~O. Hero, ``A {H}amilton-{J}acobi equation for
  the continuum limit of non-dominated sorting,'' \emph{To appear in the SIAM
  Journal on Mathematical Analysis}, 2014.

\bibitem{bollobas1988}
B.~Bollob{\'a}s and P.~Winkler, ``{The longest chain among random points in
  Euclidean space},'' \emph{Proceedings of the American Mathematical Society},
  vol. 103, no.~2, pp. 347--353, June 1988.

\bibitem{prekopa1973}
A.~Prekopa, ``On logarithmic concave measures and functions,'' \emph{Acta Sci.
  Math.(Szeged)}, vol.~34, pp. 335--343, 1973.

\bibitem{snoek2006challenge}
C.~Snoek, M.~Worring, J.~Van~Gemert, J.~Geusebroek, and A.~Smeulders, ``The
  challenge problem for automated detection of 101 semantic concepts in
  multimedia,'' in \emph{Proceedings of the 14th annual ACM international
  conference on Multimedia}.\hskip 1em plus 0.5em minus 0.4em\relax ACM, 2006,
  pp. 421--430.

\bibitem{calder2014ITA}
J.~Calder, S.~Esedoglu, and A.~O. Hero, ``{A continuum limit for non-dominated
  sorting},'' in \emph{{To appear in the Proceedings of the Information Theory
  and Applications (ITA) Workshop}}, 2014.

\bibitem{calder2013b}
J.~Calder, S.~Esedo\=glu, and A.~O. Hero, ``A {PDE}-based approach to
  non-dominated sorting,'' \emph{arXiv preprint:1310.2498}, 2013.

\bibitem{jarvelin2002cumulated}
K.~J{\"a}rvelin and J.~Kek{\"a}l{\"a}inen, ``Cumulated gain-based evaluation of
  ir techniques,'' \emph{ACM Transactions on Information Systems (TOIS)},
  vol.~20, no.~4, pp. 422--446, 2002.

\bibitem{zhou2006}
Z.-H. Zhou and M.-L. Zhang, ``Multi-instance multi-label learning with
  application to scene classification,'' \emph{Advances in Neural Information
  Processing Systems}, vol.~12, pp. 1609--1616, 2006.

\bibitem{van2006robust}
J.~van Gemert, J.~Geusebroek, C.~Veenman, C.~Snoek, and A.~Smeulders, ``Robust
  scene categorization by learning image statistics in context,'' in
  \emph{Computer Vision and Pattern Recognition Workshop, 2006. CVPRW'06.
  Conference on}.\hskip 1em plus 0.5em minus 0.4em\relax IEEE, 2006, pp.
  105--105.

\end{thebibliography}

\end{document}